%% file: template.tex
\documentclass{article}
\usepackage{arxiv}

\usepackage[utf8]{inputenc} 
\usepackage[T1]{fontenc}    
\usepackage{hyperref}       
\usepackage{url}            
\usepackage{booktabs}       
\usepackage{amsfonts}       
\usepackage{nicefrac}       
\usepackage{microtype}      
\usepackage{lipsum}		
\usepackage{graphicx}
\usepackage{natbib}
\usepackage{doi}
\usepackage{amsmath}

\usepackage{amssymb,amsfonts,amsthm,bm,bbm}
\usepackage{algorithm}
\usepackage{soul}
\usepackage{algorithmic}
\usepackage{array}
\usepackage{graphicx}
\usepackage{booktabs}
\usepackage{multirow}
\usepackage{pifont}
\usepackage{tcolorbox}
\usepackage{enumitem}
\usepackage{multicol}
\usepackage{tikz}
\usepackage{balance}

\renewcommand{\paragraph}[1]{\smallskip\noindent\textbf{#1.}}
\newcommand{\ssubsection}[1]{\subsubsection{\textbf{#1}}}

\newcommand{\proname}{\textsf{decentBRM}}

\newtheorem{definition}{Definition}

\newtheorem{remark}{Remark}
\newtheorem{Theorem}{Theorem}
\newtheorem{Claim}{Claim}
\newtheorem{Lemma}{Lemma}
\newtheorem{Assumption}{Assumption}
\setlist[itemize]{leftmargin=3mm}

\title{\textsf{Decent-BRM}: Decentralization through Block Reward Mechanisms}

\author{ \href{https://orcid.org/0000-0002-5662-0386}{\includegraphics[scale=0.06]{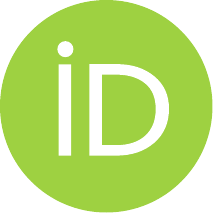}\hspace{1mm}Varul Srivastava} \\
	Machine Learning Lab\\
	IIIT Hyderabad\\
	\texttt{varul.srivastava@research.iiit.ac.in} \\
	\And
	\href{https://orcid.org/0000-0003-4634-7862}{\includegraphics[scale=0.06]{orcid.pdf}\hspace{1mm}Dr. Sujit Gujar} \\
	Machine Learning Lab\\
	IIIT Hyderabad\\
	\texttt{sujit.gujar@iiit.ac.in} \\
}

\renewcommand{\shorttitle}{decentBRM}

\hypersetup{
pdftitle={DecentBRM},
pdfsubject={q-bio.NC, q-bio.QM},
pdfauthor={David S.~Hippocampus, Elias D.~Striatum},
pdfkeywords={First keyword, Second keyword, More},
}

\begin{document}
\maketitle

\begin{abstract}
Proof-of-Work is a consensus algorithm where miners solve cryptographic puzzles to mine blocks and obtain a reward through some Block Reward Mechanism (BRM). PoW blockchain faces the problem of centralization due to the formation of mining pools, where miners mine blocks as a group and distribute rewards. The rationale is to reduce the risk (variance) in reward while obtaining the same expected block reward. In this work, we address the problem of centralization due to mining pools in PoW blockchain. We propose a two-player game between the new miner joining the system and the PoW blockchain system. 

We model the utility for the incoming miner as a combination of (i) expected block reward, (ii) risk, and (iii) cost of switching between different mining pools. With this utility structure, we analyze the equilibrium strategy of the incoming miner for different BRMs: (a) memoryless -- block reward is history independent (e.g., Bitcoin) (b) retentive: block reward is history-dependent (e.g., Fruitchains). For memoryless BRMs, we show that depending on the coefficient of switching cost $c$, the protocol is decentralized when $c = 0$ and centralized when $c > \underline{c}$. In addition, we show the impossibility of constructing a memoryless BRM where solo mining gives a higher payoff than forming/joining mining pools. While retentive BRM in Fruitchains reduces risk in solo mining, the equilibrium strategy for incoming miners is still to join mining pools, leading to centralization. We then propose our novel retentive BRM -- \textsf{Decent-BRM}. We show that under \textsf{Decent-BRM}, incoming miners obtain higher utility in solo mining than joining mining pools. Therefore, no mining pools are formed, and the Pow blockchain using \textsf{Decent-BRM} is decentralized. 

\end{abstract}


\section{Introduction}
{
}

With the increased popularity of peer-to-peer, \emph{decentralized} cryptocurrencies such as Bitcoin and Ethereum, more and more users are interested in possessing those coins/tokens. Users obtain cryptocurrency tokens by purchasing from those who mint or themself mint the tokens. The cryptocurrencies keep transaction records through a distributed ledger, and a distributed consensus protocol is required to maintain a consistent ledger across all users. \citet{nakamoto2009Bitcoin} proposed a blockchain technology via a \emph{Proof-of-Work} (PoW) consensus protocol for the same. In PoW, the interested parties can join the network and need to solve some cryptographic puzzle to \emph{mine} the next block, thereby appending a set of transactions (contained in the block) to the ledger. In return, the first node, \emph{miner}, to solve the puzzle is rewarded with newly minted coins according to some Block Reward Mechanism (BRM). 

\paragraph{Mining Pools} Due to the increasing popularity of cryptocurrencies, there has been 10x growth in computing power used for mining Bitcoins in the last five years~\cite{hashrateIncrease}, as evident from Figure~\ref{fig:hashrate-increase}. The mining of a block is a random event, and \emph{solo miners} with limited computing power face risk due to the uncertainty about mining any blocks for a prolonged duration. To minimize the risk, miners come together and form \emph{mining pools} and distribute rewards when the pool mines a block. The formation of mining pools leads to the same expected rewards but much more frequent payment, reducing the risk. Due to this, few mining pools have grown disproportionately large. E.g., the top three mining pools in Bitcoin control a majority of the computing power in Bitcoin: Foundry USA ($\approx 29\%$), Ant Pool ($\approx 16\%$) and F2 Pool ($14\%$) control a $59\%$ of the total mining power (Figure~\ref{fig:centralization-bitcoin}). The top $3$ pools in ZCash~\cite{zCash} are Nanopool ($\approx 25\%$), Mining Pool Hub ($\approx 18\%$), and Suprnova ($\approx 8\%$), control the majority ($>50\%$) of the mining power. The security of the PoW-based blockchains relies on the fact that no authority controls the majority of the computing power. Such unprecedented levels of centralization could pose a severe security threat to the PoW blockchains. 

\begin{figure}[!th]
    \centering
    \includegraphics[width=0.8\linewidth]{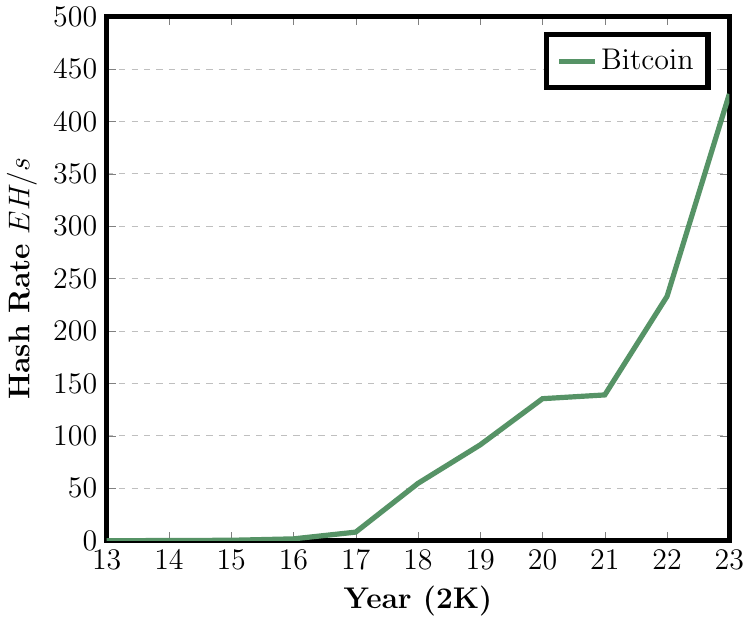}
    \caption{Bitcoin Hashrate over 10 years~\cite{empiricalHashrateIncrease}}
    \label{fig:hashrate-increase}
\end{figure}

\paragraph{Threat Due to Centralization} The security of PoW protocols relies on honest majority assumption, i.e., $>50\%$ of the mining power is controlled by honest players. However, if the majority of the mining power is concentrated with a small number of mining pools, then the protocol is under severe threat and faces the risk of censorship attacks, denial of service attacks, and security compromise, among other attacks. It would be easy to double-spend coins if the top mining pools, controlling $>50\%$ computing power, come together~\cite{Badertscher2021Bitcon51}. The challenge of $>50\%$ majority with a single mining pool is not new. For example, GHash.IO~\cite{gHashIO} controlled the majority ($55\%$) of the Bitcoin network in June 2014. Due to this, the public trust in the currency dropped. GHash.IO has since committed to limiting its mining power to $<40\%$~\cite{gHashIO}.
Thus, it is crucial to address the problem of centralization in the PoW blockchain caused by the formation of mining pools. 
Miners are rational and interested in maximizing their rewards while minimizing risks. Hence, studying the event of a new miner joining mining pools as a game is natural. 
The researchers have looked at how mining pools through game theory, focusing on optimal RSS~\cite{Roughgarden2021Variance,8}, or maximizing miner utility ~\cite{chatzigiannis2022oxfordJournal,4}. However, a limited analysis of the centralization of PoW through mining pools necessitates its investigation.

\subsection*{Our Goal} 
This work aims to study a miner joining the PoW system as a game and analyze conditions under which the system tends to be centralized or decentralized. In addition, we aim to construct a mechanism under which the formation of mining pools is not profitable over solo mining, and the system remains decentralized. 

\paragraph{Our Approach} To analyze the behavior of a miner joining the PoW blockchain system with different block-reward mechanisms (BRMs) represented by $\Gamma$. Towards this, we propose a two-player game $\mathcal{G}(\Gamma)$ for BRM $\Gamma$ played between player $p_{1}$ deciding to join the system, and the PoW blockchain system represented as player $p_{2}$. We define the utility of $p_{1}$, which depends on (i) expected reward, (ii) risk, and (iii) function switching costs. (i) Expected rewards are a function of computing power. (ii) Since variance is an indicator of risk~\cite{chiu2016RiskNobel}, miners intend to reduce this risk. We model it via $\rho^{th}$ moment of the reward, which is a more general representation of risk (in literature, it is usually represented as the $2^{nd}$ moment, i.e. variance~\cite{chiu2016RiskNobel,Roughgarden2021Variance}). (iii) modeling switching costs are needed to account for the cost in hash queries, network latency, etc., in switching from one mining pool to another during the mining process (Section~\ref{ssec:utility}). 
 We also define the conditions of (i) Fairness (motivated by conditions of unbiased Reward Sharing Scheme of~\cite{Roughgarden2021Variance}) and (ii) $\rho$-Decentralization. 
 To analyze $\mathcal{G}(\Gamma)$, we categorize BRMs into two types, \emph{memoryless} and \emph{retentive}. Suggestive by the nomenclature, memoryless BRMs (eg. Bitcoin~\cite{nakamoto2009Bitcoin}) are such that the reward given for any block is independent of the ledger's history and retentive BRMs employ use of ledger's history while distributing reward (eg. Fruitchain~\cite{Rafael2017Fruitchains}). 

\begin{figure}[!th]
    \centering
    \includegraphics[width=0.8\linewidth]{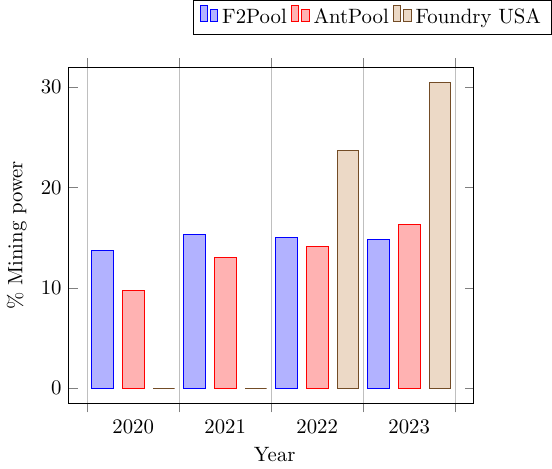}
    \caption{Centralization in Bitcoin over 4 years~\cite{bitcoinStatistics}}
    \label{fig:centralization-bitcoin}
\end{figure}


\subsection*{Our Contributions}
First, we study memoryless BRMs and model switching costs through a parameter $c$ and function $D$. If $c=0$ (i.e., it does not assign any switching cost ), player $p_1$ distributes the mining power across mining pools, implying decentralization is possible. However, we show an impossibility result that no memoryless BRM can assure $\rho$-decentralization if switching costs exceed a certain threshold. Hence, we turn to retentive BRMs (eg. Fruitchains~\cite{Rafael2017Fruitchains}), where the reward given for any block depends on historical entries in the ledger. Our analysis shows that Fruitchain reduces risk and increases utility for solo mining. However, joining mining pools is still a better strategy than solo mining (because mining pools also experience reduced risk). Despite the above negative results, we show that retentive BRMs increase the utility of solo mining. Taking motivation from this, we propose our novel BRM \proname. In \proname, if a block is mined, the reward is distributed among all the miners who have mined blocks before the current block. This protocol is such that following the solo mining strategy is an equilibrium strategy (Lemma~\ref{lemma:proto-decentralized}). Therefore, the PoW blockchain protocol using $\Gamma_{\proname}$ is $\rho$-Decentralized. 

In summary, the following are our contributions:
\begin{itemize}
    \item We define a two-player game to study these BRMs where player $p_{1}$ is the miner joining the system and $p_{2}$ is the rest of the system, comprising of mining pools and different solo miners. (Section~\ref{ssec:players})
    \item In Section~\ref{ssec:game-properties}, we introduce the notion of  $\rho$-Decentralization. 
    \item  We prove two results for memoryless BRMs (i) for $c = 0$ protocol is $\rho$-Decentralized (Lemma~\ref{thm:2-decentralization}) and (ii) for $c \geq \underline{c}$, centralization is bound to happen; where $\underline{c}$ is a threshold based on the other parameters (Theorem~\ref{thm:bakwaas}).
    \item We also show through Theorem~\ref{thm:impossibility} that it is impossible to construct a memoryless BRM where solo mining is incentivized over joining mining pools.
    \item We propose a novel retentive BRM, namely \proname, that incentevizes solo mining over joining mining pools; thus, it is $\rho$-Decentralized. This helps in addressing the centralization of computing power in PoW blockchains.
\end{itemize}

\section{Related Work}

Game Theory is extensively employed for Blockchains security analysis~\cite{blockchainPoA,DavideAAMAS22,kiayiasNash,AJAINAAMAS}, constructing Transaction Fees mechanism~\cite{roughgardenTFM1,elaineShiTFM,bitcoinF}, Consensus~\cite{BlockchainGT}, low-variance mining~\cite{Roughgarden2021Variance,8} et cetera. 

\paragraph{Centralization in PoW \& Mining pools} Centralization in PoW blockchains poses a major threat to protocol security~\cite{3,14,15}. \citet{5} argues that centralization in PoW blockchains is due to high variance in reward for mining a block. As shown by~\citet{4}, the creation of mining pools reduces this variance in the reward, which miners will be willing to join even at (some limited) expense of their expected reward. Mining pools offer different reward schemes to the miners, such as the Pay-Per-Share (PPS), Proportional, Geometric, and Pay-Per-Last-N-Shares (PPLNS). \citet{Roughgarden2021Variance} formally studied \emph{reward sharing schemes} (RSS) opted by different mining pool strategies. While all RSS have the same expected payoff, they offer different risks (measured as variance in reward). Roughgarden and Shikhelman show that Pay-Per-Share (PPS) and Pay-Per-Last-N-Shares (PPLNS) are variance-optimal RSS under different conditions. \citet{8} also study reward schemes offered by mining pools from a Game Theoretic lens and find mining pool strategies that maximize optimal steady state social welfare for miners. However, these works~\cite{Roughgarden2021Variance,8} analyze reward-sharing schemes while our work involves analysis of Block reward mechanisms (BRMs) and the possibility of decentralization through them in the PoW system by disincentivizing the creation of mining pools. The role of BRMs for the security of blockchain~\cite{15} and security and decentralization in PoS blockchain~\cite{1} have already been studied. We study the role of BRMs in decentralizing PoW blockchains, which indirectly ensures security~\cite{Badertscher2021Bitcon51,14} of the protocol. 

\paragraph{Game Theoretic analysis of Miners} Game theoretic analysis of miner behaviour under PoW blockchains has also been performed previously through Mean Field Games~\cite{16}, Evolutionary Game Theory~\cite{19} studying equilibrium under different conditions of computational power and network propagation delay. However, these works lack a general model capturing miner payoffs, risk and additional costs (such as switching costs), which are captured by our work. In addition, these works analyze the miner strategy to choose among different reward-sharing schemes under Bitcoin. In contrast, our work analyses the miner choice among different reward-sharing schemes under different block reward mechanisms, including (but not limited to) Bitcoin. Similarly, \citet{18} studies miner's dynamic choice among different mining pools. However, their utility structure is restricted to expected reward and does not account for risk (eg. variance) among other costs. Multiple other works analyze miner choice in joining pools based on reward sharing schemes, network delay, and cost of mining, such as~\cite{17,Roughgarden2021Variance,8,4,1}. Note that this is not an exhaustive list. \citet{chatzigiannis2022oxfordJournal} propose portfolio strategies for miners with different risk aversion levels to join across different mining pools and even different currencies. However, their analysis is completely empirical based on their computational tool. The goal of their experimental analysis is to find portfolio strategies which maximize risk-adjusted returns for miners, and in no way discuss and resolve centralization in PoW blockchains due to pool formation.

\paragraph{Protocols for Decentralized Mining}
There also have been efforts to introduce mechanisms that reduce miner variance such as modification to Bitcoin~\cite{20}, Bobtail~\cite{23}, Fruitchains~\cite{Rafael2017Fruitchains}, Hybrid PoW protocols~\cite{21}, SmartPool~\cite{21-11}, Proof-of-Mining~\cite{6}, Sign-to-Mine~\cite{21-6}, Multi-nonce schemes~\cite{24}, HaPPY-Mine~\cite{2} among others. However, they all reduce the incentive to join mining pools or lead to small-sized mining pools. In any case, pool formation is still incentivized, and the PoW system is centralized. \citet{Stouka2023FruitchainsAnalysis} show that Fruitchains~\cite{Rafael2017Fruitchains} reduce variance for solo mining and mining pools. Hence, it is centralized. We make similar observations in our analysis of fruitchains.

\section{Preliminaries}
In this section, we explain (i) blockchain preliminaries, (ii) block reward mechanisms, and (iii) the functioning of mining pools.

\subsection{Blockchain Preliminaries}
\ssubsection{Proof-of-Work (PoW)} A \emph{blockchain} is a ledger maintained and updated by a set of \emph{miners} which we also refer to as \emph{players}. The blockchain comprises an ordered chain of blocks, where a block at height $k$ is $B_{k}$ and is preceded by a ``parent-block'' $B_{k-1}$ for all $k>0$\footnote{$k=0$ is genesis block}. A block $B_{k}$ consists of block-header and transaction data. The header contains information (1) hash (cryptographic digest) of the parent block $B_{k-1}$ (Null genesis block), (2) height of the block $k$, (3) Merkle-root (cryptographic fingerprint) of the set of transactions included in the block, and (4) a random number $n$, called \emph{Nonce}. Note that the block structure may vary among different PoW blockchains, but these four elements are invariably present in the block header and are sufficient for our understanding. The protocol decides on some \emph{Target} $T$ and considers a block \emph{mined} if its header's hash is lesser than the target $T$.
Thus, each miner searches for a nonce till the block is mined by investing in computing power to compute hashes. The \textbf{hash rate} of a miner is how many queries it can make in a time interval. Any \textbf{round} $k$ is defined as the duration from after block $B_{k-1}$ is mined till when block $B_{k}$ is mined. The entire protocol is parameterized by the security parameter $\lambda$ which is also the size of the Hash function.

\ssubsection{Block Reward Mechanism} At round $k$, a PoW blockchain protocol rewards a miner for its work if it successfully mines the block $B_k$. Typically, the rewards are of two types: (i) \emph{block reward} and (ii) \emph{transaction fees}. We abstract out the \emph{block-reward mechanism} (BRM) of the PoW blockchain for a block $B_{k}$ mined in round $k$ based on the history $\mathcal{H}_{k}$ at round $k$ as $\Gamma(B_{k},\mathcal{H}_{k})$. At round $k$, for most of the PoW blockchains, BRMs offer block rewards to the miner independent of the chain's evolution till round $k$. We refer to such BRM as \emph{memoryless BRM}. We formally define the memorylessness condition below.

\begin{definition}[Memoryless BRM]\label{def:memoryless}
    A BRM $\Gamma(\cdot,\cdot)$ has \emph{memoryless} property if for any round $k$ and any history $\mathcal{H}^{1}_{k},\mathcal{H}^{2}_{k}$ we have 
    \begin{equation}\label{eqn:memoryless}
        \Gamma(B_{k},\mathcal{H}^{1}_{k}) = \Gamma(B_{k},\mathcal{H}^{2}_{k})
    \end{equation}
\end{definition}

A trivial BRM which is memoryless is $\Gamma(B_{k},\mathcal{H}_{k}) = 0$ for all $B_{k},\mathcal{H}_{k}$. Another example of memoryless BRM is bitcoin's BRM ($\Gamma_{memoryless}$). $\Gamma_{memoryless}$ rewards only the miner who has mined the current block $B_{k}$ and is therefore invariant of the previous blocks mined (history). Reward in a memoryless BRM might depend on the block's height. (In Bitcoin, the block reward halves after every 210K blocks are mined.) More sophisticated BRMs can leverage the miners' historical activity and pay more to the miners working harder. Examples of such retentive (or non-memoryless) BRMs are Fruitchains~\cite{Rafael2017Fruitchains} and \proname\ (protocol proposed in Section~\ref{sec:protocol}).

In memoryless BRMs, miners face much uncertainty in mining the block due to the increased competition. Instead of solo mining, often multiple miners come together for mining and share the rewards amongst themselves; though expected rewards remain the same, the risk (variance in the rewards) decreases. For example, suppose a solo miner starts mining; there is a 10\% chance of mining a block with a reward of 10 coins. If two friends with the same computing power decide to mine jointly and split the rewards equally, the expected reward is still one coin each. However, now, the chance of `0' reward is 81\% as opposed to 90\% in solo mining. Thus, PoW mining leads to \emph{mining pools}.

\subsection{Mining Pools} In PoW blockchains, although the expected reward for each miner is sufficiently adequate, the frequency at which each miner gets rewarded (through mining a block) is very low. Multiple miners work together to mine a block. All these miners together are called a \emph{mining pool}. As the pool has a higher mining power than individual miners, a higher chance of mining block in any round. The pool is managed by a pool manager and is constituted of multiple miners. We assume the PoW blockchain under analysis has $\mathcal{P}=\{1,2,\ldots,p\}$ mining pools. Following is the working of mining pool $i$.
\begin{itemize}
    \item Pool manager constructs the skeleton of the block the pool will attempt to mine. If the block is mined, the reward is transferred to the pool (with public key $pk_{i}$). 
    \item Each miner tries different nonces to solve the difficulty puzzle. They send the results of some (or all) queries they made, which serve as input to the \emph{reward-sharing scheme} (RSS) opted by the pool for when a block is mined. 
\end{itemize}
We now describe the message transfer within the mining pool.

\ssubsection{Message Distribution} Each mining pool comprises different miners. Miners send messages to the mining pool $i$, which usually comprises of the miner identity (public key) $s$ of the miner and uniformly samples binary string $n \in \{0,1\}^{\lambda}$ (nonce for the block header) such that the hash of the block is $< T_{i}$ for some target decided by the pool. This target is easier than the PoW blockchain target (i.e. $T_{i} > T$). Each message belongs to the message set $\mathcal{M}$. We represent the set of messages obtained in a round $k$ by pool $i$ as $\mathcal{M}_{k}^{i}$. We represent the message as the tuple $m:= (s,n)$. 

\smallskip \noindent RSS have been well studied in the literature~\cite{Roughgarden2021Variance,fisch2017WineMiningPools}. Our work abstracts out the expected reward for analysis and is RSS agnostic\footnote{under certain conditions due to Claim~\ref{claim:fair-rss}}. 

\ssubsection{Reward Sharing Scheme~\cite{Roughgarden2021Variance}} Reward Sharing Scheme (RSS) is a function that takes in the set of messages received by a mining pool, along with the identity of a miner $s$ which is part of the pool, and outputs what fraction of the reward should be given to the miner. The set of messages received by a mining pool $i$ in round $k$ is $\mathcal{M}_{i}^{k}$. Then, for any miner $s$ that is part of the pool, the RSS for pool $i$ is a function $\psi_{i}(s,\mathcal{M}_{i}^{k}) \in [0,1]$. The reward given to the player $s$ is therefore, $\Gamma^s_k:=\Gamma(B_{k},\mathcal{H})\psi_{i}(s,\mathcal{M}_{i}^{k})$. If we are analyzing reward for a single player, then $\psi:\{s,s'\}\times\mathcal{M} \rightarrow [0,1]$ where $s$ indicates the fraction of reward given to the miner, and $s^{'}$ represents the fraction given to rest of the miners. Details on Message distribution and reward sharing scheme can be found in~\cite{Roughgarden2021Variance}.

As the introduction explains, we want to study centralization in mining pools. Few mining pools control most of the mining power, posing security concerns for the PoW blockchains. We assume all the miners are strategic players interested in maximizing their rewards and minimize the risk. Thus, we model problem of a player who want to join mining as a game in the next section.

\section{The Game}\label{sec:game}
We consider the blockchain system with $p$ mining pools as $\textsf{M} = \left(\Gamma,\psi\right)$ which consists of $\Gamma$ -- block reward mechanism and $\psi_{i}$ -- reward sharing scheme opted by the mining pool $i$.
As there is a large number of miners already in the system, we model the dilemma of a new player about how to distribute its computing power as a two-player game $\mathcal{G}(\Gamma)$ for BRM $\Gamma$. Player $p_{1}$  is the miner joining the system and $p_{2}$ abstracts the remaining system (comprising of mining pools and solo miners) into a single player. 

In this section, (i) we explain different classes of player 1 -- discussing risk tolerance of $p_{1}$, (ii) the strategy space of the players, (iii) the utility structure of $p_{1}$ that is joining the PoW system, and (iv) the desired properties for PoW system -- fairness and decentralization.


\subsection{Players}\label{ssec:players}
The Game consists of players $\mathcal{P} = \{p_{1},p_{2}\}$. Player $p_{1}$ is characterized by $\theta_{1} \in \mathbb{R}_{>0}\times\mathbb{N}$ represented as $\theta_{1} = (M_{1},\rho)$ where $M_{1}$ is the total hash rate of player $p_{1}$ and $\rho$ is its risk-tolerance. The player $p_{2}$ has type $\theta_{2} \in \mathbb{R}_{> 0}\times\mathcal{A}$ where $\mathcal{A}$ is the set of possible reward sharing schemes followed by different mining pools. 
%
We make the following assumption wherever necessary: owing to the enormous mining power of the PoW blockchain system compared to any single miner~\cite{bitcoinStatistics}. 

\begin{Assumption}\label{assumption:approximate-mining-power}
        For player $p_{1}$ characterized by $\theta_{1} = (M_{1},\rho)$ and $p_{2}$ by $\theta_{2} = (M_{2}, A)$, the mining power of a single player is minimal compared to the rest of the system i.e. $M_{1} << M_{2}$.
\end{Assumption}

\ssubsection{Risk Tolerance} 
Let $R_{k}$ be the random variable representing player $p_{1} $'s reward in round $k$. The primary reason for $p_1$ to join one or more mining pools is minimizing the risk $R_k$. Solo mining 
has a very high variance (the second-order moment of $R_k$) in the reward obtained.  According to~\cite{chiu2016RiskNobel}, \texttt{
``Pioneered by Nobel laureate Harry Markowitz in the 1950s, the mean-variance (MV) formulation is a fundamental theory for risk management in finance.''}. 
MV is the difference between the mean reward and its standard deviation -- the risk. Instead of just modelling risk through second moment of the reward, we parameterize it by a tolerance of player $p_{1}$ using $\rho \in \mathbb{N}$. 
A $\rho$-risk tolerant player will want to minimize $\left(\mathbb{E}[R_{k}^{\rho}]\right)^{1/\rho}$. A very high $\rho$ is characteristic of risk-averse whereas $\rho=1$ is a risk-neutral. We do not consider risk-loving players as they would prefer solo mining.

\ssubsection{Strategy Space and Environment}
Each player follows a strategy sampled from the set $\mathcal{S} = \Delta(p+1)$ which is a $(p+1)-$simplex; $\overline{g}$ for player $p_1$ and $\overline{f}$ for player $p_2$. $\overline{g}=(g_{0},g_{1},\ldots,g_{p})$ such that $\sum_{i=0}^{p} g_{i} = 1$. $g_i$ indicates the fraction of mining power used to mine with mining pool $i$, and $g_0$ indicates the fraction used for solo mining. Similarly, $\overline{f}$ is the strategy for player $p_2$. 
Therefore, the hash rate for pool $i$ is given by $M_{2}\cdot f_{i}$ at the time $p_1$ is joining.
Thus, the total of $M_{2}f_{0} + M_{1}g_{0}$ computing power -- solo miners use hash rate, and for each mining pool $i$, the hash rate is $M_{2}f_{i} + M_{1}g_{i}$. As $M_2 >> M_1$, we assume that the joining of a single player $p_1$ does not affect $\overline{f}$. 

\paragraph{Discrete Strategy Space} The strategy space described above captures that a player can reallocate some of its mining power across different mining pools. However, practically, there will be some limit on the minimum computing power it can assign to a particular pool. Hence, we discretize the strategy space based on this minimum possible change. We define a strategy space with minimum permissible change being $\alpha$ as $\alpha$-\emph{Discrete Strategy Space}, formally defined in Definition~\ref{def:marginal-change} below. 
\begin{definition}[$\alpha$-discrete Strategy Space]\label{def:marginal-change}
    A strategy space $S_{\alpha}$ is $\alpha$-discrete iff $\forall g^{1},g^{2} \in S_{\alpha},\forall i \in \{0,1,2,\ldots,k\}\;\exists l \in \mathbb{Z}_{\geq 0}$
    \[
        |g^{1}_{i} - g^{2}_{i}| = l\cdot \alpha  
    \]
\end{definition}

\subsection{Utility of a Player}\label{ssec:utility}
Having defined players and the strategy space, we now define the utility for player $p_{1}$. Towards this, consider the random variable $R_{k}$ representing the reward obtained by the player in round $k$. If $\Gamma(B_{k},\mathcal{H}_{k})$ is the (PoW protocol specific) BRM in a round $k$ and $\psi_{i}:\{s_{1},s_{-1}\}\times\mathcal{M}\rightarrow[0,1]$ is the Reward Sharing Scheme (RSS) opted by pool $i$. Then, the reward in round $k$ is given by random variable $R_{k}$ (randomness is due to distribution $\mathcal{M}(\overline{g},\overline{f})$) as: 
\begin{equation}\label{eqn:reward}
    R_{k} = \begin{cases} 
      \Gamma(B_{k},\mathcal{H}_{k}) & \text{with prob. }\frac{g_{0}M_{1}}{M_{1} + M_{2}} \\
      \Gamma(B_{k},\mathcal{H}_{k})\psi_{i}(s_{1},\mathcal{M}_{i,r}(\overline{g},\overline{f})) & \text{with prob. }\frac{g_{i}M_{1} + f_{i}M_{2}}{M_{1} + M_{2}} \\
      0 & \text{with prob. }\frac{f_{0}M_{2}}{M_{1} + M_{2}} \\
   \end{cases}
\end{equation}

\ssubsection{Message Distribution} Miners who are part of some mining pool send messages to the pool during mining attempts. If the message was sent by player $p_{1}$, then public key $s = s_1$ and $s = s_{-1}$ otherwise. These messages are random variables $m \in \{s_{1},s_{-1}\}\times\{0,1\}^{\lambda}$ induced by the hash rate distribution $\overline{g},\overline{f}$ and is therefore represented as $\mathcal{M}(\overline{g},\overline{f})$. The distribution is such that for any randomly sampled message $m \in \mathcal{M}^{i}_{k}(\overline{g},\overline{f})$ in any arbitrary round $k$, we have:
\begin{equation*}    
m = \begin{cases} 
      (s_{1},n) & \text{with prob. }\frac{g_{i}M_{1}}{M_{1}g_{i} + M_{2}f_{i}} \\
      (s_{-1},n) & \text{with prob. }\frac{f_{i}M_{2}}{M_{1}g_{i} + M_{2}f_{i}} \\
      \end{cases}
\end{equation*}
Here, $n$ is the solution to the puzzle (nonce), and $\lambda$ is the size of the hash function output, also a security parameter of the protocol.

\paragraph{Pool Hopping / Switching Cost} Player $p_1$ typically may hop between different pools (called pool-hopping~\cite{Cortesi2022PoolHopping}), to minimize risks. However, changing the pool incurs some costs, e.g., a drop in the effective hash rate. We model such pool hopping costs as \emph{switching cost function }$D:\Delta(K+1)\rightarrow\mathbb{R}_{\geq 0}$. $D$ has the following properties: 
\begin{itemize}
    \item[\textbf{P1}.] $D(\overline{g})$ monotonically increases with $\left|\{ g_{i} | g_{i} \neq 0; i \in \{0,1,2,\ldots,k\}\}\right|$.
    \item[\textbf{P2}.] For $\overline{g}^{1},\overline{g}^{2} \in S_{\alpha}$, if conditions (i) $g^{1}_{q} = g^{2}_{q} \forall q \in[1,k]/\{i,j\}$, (ii) $g^{1}_{i} + g^{1}_{j} = g^{2}_{i} + g^{2}_{j}$, and (iii) $g^{1}_{i}\cdot g^{1}_{j} > g^{2}_{i}\cdot g^{2}_{j}$ holds, then $D(\overline{g}^{1}) \leq D(\overline{g}^{2})$
\end{itemize}

The property \textbf{P1} states if a player hops across more pools, the switching cost increases as it needs to switch more frequently. The property \textbf{P2} corresponds to higher switching cost (due to more frequent switching between mining pools) if mining power is distributed more equally between two (or more) pools. In addition to this, if the strategy space is $\alpha$-discrete, then the switching cost changes by at least some minimum amount $D_{min}$ between two strategies. We call this Marginal Switching Cost and is formally defined in Definition~\ref{def:marginal-switching-cost}. 
\begin{definition}[Marginal Switching Cost]\label{def:marginal-switching-cost}
    For switching cost function $D(\cdot)$ and strategies $g^{1},g^{2} \in S_{\alpha}$ (when $D(\overline{g}^{1}) \neq D(\overline{g}^{2})$) where $S_{\alpha}$ is an $\alpha$-discrete Strategy space, the \emph{Marginal cost} of switching $D_{min}>0$ is 
    \[
        D_{min} := \min_{g^{1},g^{2} \in S; g^{1} \neq g^{2}} \left|D(g^{1}) - D(g^{2})\right|
    \]
\end{definition}

In summary, player $p_{1}$'s utility constitutes reward (positively) and risk and switching costs (negatively). We weigh them with $a,b,c$, respectively. Thus, in a round $k$ its utility is: 
\begin{equation}\label{eqn:utility-single-round}
    u_{1}\left(\overline{g},\overline{f},r;(\theta_{1},\theta_{2})\right) = a\cdot\mathbb{E}[R_{k}] - b\cdot\left({\mathbb{E}[R_{k}^{\rho}]}\right)^{1/\rho} - c\cdot D(\overline{g})
\end{equation}
Considering a discount factor $\delta \in [0,1]$ the utility of $p_{1}$ across rounds (starting from round $r_{0}$) is given by 
\begin{equation}\label{eqn:utility-across-rounds}
    U_1\left(\overline{g},\overline{f};(\theta_{1},\theta_{2})\right) = \sum_{r=r_{0}}^{\infty} \delta^{r-r_{0}} u_{1}\left(\overline{g},\overline{f},r;(M_{1},\rho,M_{2})\right)
\end{equation}

We show the equilibrium strategy for player $p_1$ as well as fairness and decentralization defined in the next subsection are agnostic to player $p_2$'s utility $U_2$ (see Remark~\ref{rem:utility-agnostic}). Thus, we do not need to define it and hence, we skip it.  

\subsection{Game Progression and Properties}\label{ssec:game-properties}
\label{ssec:properties}
The goal of this game is to model the event of player $p_{1}$ joining the system and following utility maximizing strategy $\overline{g}^{*} \in S_{\alpha}$. 

\ssubsection{Game} Our game progresses in two steps. First, player $p_{2}$ plays, disclosing its opted strategy $\overline{f}$ following which, player $p_{1}$ plays strategy $\overline{g}$ such that $U_1\left(\overline{g},\overline{f};(\theta_{1},\theta_{2})\right)$ is maximized. Note that $p_{1}$ observes $\overline{f}$ before playing $\overline{g}$. Thus,  $\overline{g}$ can be a depends of $\overline{f}$. For notational ease, we slightly exploit notations and represent $\overline{g}_{\overline{f}}$ as just $\overline{g}$ unless necessary. Therefore optimal strategy $\overline{g}^*$ for a given strategy $\overline{f}$ is solution to the optimization problem for $p_{1}$ (after observing $\overline{f}$) given by: 
\[
    \begin{aligned}
    \max_{\overline{g}}\hspace{10pt} & U_{1}\left(\overline{g},\overline{f};((M_{1},\rho),(M_{2},A))\right) \\ 
    \text{s.t.}\hspace{10pt} & g_{i} \geq 0 \forall i \in \{0,1,\ldots,p\}\\
    & \sum_{i \in \{0,1,\ldots,p\}} g_{i} = 1\\
    \end{aligned}
\]

\begin{definition}[Equilibrium Strategy]\label{def:eq}
    For a BRM $\Gamma$, and game $\mathcal{G}(\Gamma)$, the strategy $\overline{g}^{*}$ is Equilibrium strategy for $p_{1}$ if $\forall \theta_{1} \in \mathbb{R}_{> 0}\times\mathbb{N}, \forall\theta_{2} \in \mathbb{R}\times\mathcal{A}, \forall \overline{g}^{'}\in S_{\alpha}$ and $\forall \overline{f}$ we have 
    \[
    U_{1}\left(\overline{g}^{*},\overline{f};(\theta_{1},\theta_{2}) \right) \geq U_{1}\left(\overline{g}^{'},\overline{f}; (\theta_{1},\theta_{2})\right) 
    \]
    We call strategy $\overline{g}^{*}$ as Dominant Strategy for the given $\overline{f}$.
\end{definition}

Based on the outcome of the game, we consider the following properties which are required in our analysis.

\ssubsection{Fairness} The reward sharing scheme for a mining pool $i$ is fair for player $p_{1}$ if its reward sharing scheme serves as a hash rate estimator. This notion of fairness is motivated by the definition of unbiased RSS in~\cite{Roughgarden2021Variance}. For completeness, we define fairness wrt. our setting in Definition~\ref{def:fairness}. 
\begin{definition}[Fairness~\cite{Roughgarden2021Variance}]\label{def:fairness}
    If mining pool $i$ (controlling $f_{i}$ fraction of total mining power $M_{2}$) follows Reward Sharing Scheme $\psi_{i}:\{s_{1},s_{-1}\}\times\mathcal{M}^{*}\rightarrow[0,1]$, then given mining pool is fair for player $p_{1}$ following strategy $\overline{g}$ if 
    \begin{equation}\label{eqn:fairness-condition}
        \mathbb{E}[\psi_{i}(1,\mathcal{M}^{i}_{k})] = \frac{g_{i}M_{1}}{g_{i}M_{1} + f_{i}M_{2}}
    \end{equation}
\end{definition}

We show through the Claim~\ref{claim:fair-rss} below that a mining pool should have a fair RSS to incentivize other miners to join the pool. 

\begin{Claim}
\label{claim:fair-rss}
    If a mining pool $i$ follows a RSS $\psi_{i}$ that is not fair, then for any BRM $\:\Gamma$ and corresponding game $\mathcal{G}(\Gamma)$, any strategy $\overline{g}^{2}$ with $g^{2}_{0} = x$ and $g^{2}_{i} = y$ is dominated by $\overline{g}^{1}$ with $g^{1}_{0} = x+y$ and $g^{1}_{i} = 0$ for $p_{1}$.
\end{Claim}

\ssubsection{Decentralization} The mechanism $\left(\Gamma,\psi\right)$ is comprised of block reward scheme $\Gamma$ and reward sharing scheme $\psi$. We define decentralization of a mechanism as $\rho-$Decentralized when miners joining the system have risk tolerance $\rho$. We call a Mechanism $\rho-$\emph{Decentralized} if the relative mining power of the highest mining pool does not change upon a strategic player joining the system. 

\begin{definition}[$\rho-$Decentralized] \label{def:weak-decentralization}
    Consider a mechanism $\left(\Gamma,\psi\right)$ is $\rho-$Decentralized for game $\mathcal{G}(\Gamma)$ played between player $p_{1}$ and PoW blockchain system $p_{2}$ where:
    \begin{itemize}
        \item $p_{1}$ characterized by $\theta_{1} = (M_{1},\rho)$ follows equilibrium strategy $\overline{g}^*$ on observing $\overline{f}$.
        \item $p_{2}$ characterized by $\theta_{2} = (M_{2},A)$ follows any arbitrary strategy $\overline{f}$ 
    \end{itemize}
    That is, $\overline{g}^{*}$ satisfies
    \begin{equation}
        \max_{i \in \{1,2,\ldots,k\}} f_{i} \geq \max_{j \in \{1,2,\ldots,k\}} \frac{g^{*}_{j}M_{1} + f_{j}M_{2}}{M_{1} + M_{2}}
    \end{equation}
\end{definition}


\begin{remark}\label{rem:decreasing-decentralization}
    If a mechanism is $\rho-$Decentralized then for every player joining the system, the value of centralization decreases (or remains the same). Therefore, the mechanism ensures that the system decreases (or at least preserves) centralization in the PoW system.
\end{remark}

The remark follows from the definition of Decentralization (Definition~\ref{def:weak-decentralization}). Consider miners $m_{1},m_{2},m_{3}...$ joining the system leads to $G_{j} = \max_{i \in \{1,2,\ldots,k\}} f_{i,j}$ where $f_{i,j}$ is the mining power of the pool $i$ after miner $m_{j}$ joins the system. Therefore, if a system is $\rho-$Decentralized, then after miners $m_{1},m_{2},m_{3},\ldots$ with risk tolerance $\rho$ join the system, by Definition~\ref{def:weak-decentralization}, we have $G_{1} \geq G_{2} \geq G_{3} \geq \ldots$. Therefore, centralization always decreases or remains the same.

\begin{remark}\label{rem:utility-agnostic}
    Fairness and Decentralization guarantees provided by mechanism $\left(\Gamma,\psi\right)$ are agnostic to the utility function $U_{2}(\cdot)$ of $p_{2}$.
\end{remark}

Remark~\ref{rem:utility-agnostic} follows from the observation that fairness property does not require any constraints on utility and decentralization property should be satisfied for all strategies followed by $p_{2}$. 

In summary, the game used for analyzing centralization due to mining pools is defined as $\mathcal{G} := \langle\{p_{1},p_{2}\},\{\theta_{1},\theta_{2}\},S_{\alpha},\Gamma,\mathcal{M},(U_{1},U_{2})\rangle$. In the following section, we use this definition of $\mathcal{G}$ to analyse conditions that lead to centralization or decentralization of a PoW blockchain system with some BRM $\Gamma$. 

\section{Theoretical Analysis}\label{sec:analysis}

This section provides analysis of (i) memoryless BRMs like Bitcoin~\cite{nakamoto2009Bitcoin} ($\Gamma_{memoryless}$), (ii) retentive (non-memoryless) BRMs such as Fruitchains~\cite{Rafael2017Fruitchains} ($\Gamma_{fruitchains}$). In this analysis, we find conditions under which these mechanisms lead to centralization and decentralization of the system. We use bitcoin reward function for memoryless BRMs because they abstract out the 

\subsection{Memoryless BRMs ($\Gamma_{memoryless}$)}\label{ssec:memoryless}

Bitcoin is a well-known PoW-based blockchain protocol, which uses memoryless BRM (Definition~\ref{def:memoryless}). Consider indicator function $\mathbb{I}$ for block $B_{k}$ and player with public-key $s$ as $\mathbb{I}(s,B_{k}) = 1 $ if the block is mined by player $s$ and $0$ otherwise. $\Gamma_{memoryless}$ in round $k$ for player (or mining pool) with public key $s$ is
\[
    \Gamma_{memoryless}^{s_{1}}(B_{k},\mathcal{H}_{k}) = \mathbb{I}(s_{1},B_{k})R_{block} 
\]
For notational ease, we write the BRM $\Gamma_{memoryless}^{s_{1}}(B_{k},\mathcal{H}_{k})$ as $\Gamma_{memoryless,k}^{s_{1}}$. If the block is mined by the player $p_{1}$ or one of the $p$ mining pools, $p_{1}$ gets rewarded proportionately as given by Equation~\ref{eqn:reward}. The game using Bitcoin-like memoryless BRMs $\Gamma_{memoryless}$ is $\mathcal{G}(\Gamma_{memoryless}) := \langle\{p_{1},p_{2}\},\{\theta_{1},\theta_{2}\},S_{\alpha},\Gamma_{memoryless},\mathcal{M},(U_{1},U_{2})\rangle$. Consider player $p_{1}$ has public-key $s_{1}$ and mining pool $i$ has public key $pk_{i}$. From Equation~\ref{eqn:reward}, the reward for $p_{1}$ in $\mathcal{G}(\Gamma_{memoryless})$ (after Assumption~\ref{assumption:approximate-mining-power}) is
\[
    R^{memoryless}_{k} = \begin{cases}
      \Gamma_{memoryless,k}^{s_{1}} & \text{with prob. }\frac{g_{0}M_{1}}{M_{2}} \\
      \Gamma_{memoryless,k}^{pk_{i}}\psi_{i}(s_{1},\mathcal{M}_{i,r}(\overline{g},\overline{f})) & \text{with prob. }f_{i} \\
      0 & \text{with prob. }f_{0} \\
    \end{cases}
\]

Due to Claim~\ref{claim:fair-rss}, mining pools follow a Fair Reward Sharing Scheme. Therefore, for pool $i$
\[
    \psi_{i}\left(s_{1},\mathcal{M}_{i,k}(\overline{g},\overline{f})\right) = \frac{g_{i}M_{1}}{g_{i}M_{1} + f_{i}M_{2}} \approx \frac{g_{i}M_{1}}{f_{i}M_{2}}
\]
Therefore, the utility for $p_{1}$ from Equation~\ref{eqn:utility-single-round} is given as
\[
\begin{aligned}
    U_{1}\left(\overline{g},\overline{f};((M_{1},\rho),(M_{2},A))\right) \approx & a\sum_{i=1}^{p} \frac{g_{i}M_{1}}{M_{2}}R_{block} \\ & - b\cdot\left(R^{\rho}_{block}\cdot\sum_{i=1}^{k} \frac{g_{i}^{\rho}}{f_{i}^{\rho-1}}\right)^{1/\rho}\\ 
    & - c\cdot D(\overline{g})
\end{aligned}
\]
We first show through Lemma~\ref{thm:2-decentralization} that a PoW system with incoming player $p_{1}$ characterized by $\theta_{1} = (M_{1},\rho=2)$ (i.e. variance is considered as the risk) is $\rho-$Decentralized for memoryless BRMs.
\begin{Lemma}\label{thm:2-decentralization}
    A PoW blockchain with memoryless BRM $\Gamma_{memoryless}$ is $\rho-$Decentralized under $\mathcal{G}(\Gamma_{memoryless})$ for any $\rho > 1$ if $c=0$.
\end{Lemma}
\begin{proof}
    The proof follows by showing that given a strategy $\overline{f}$ followed by $p_{2}$ (which is observable by $p_{1}$ in $\mathcal{G}(\Gamma_{memoryless}$), the player $p_{1}$ following $\overline{g} = \overline{f}$ gives higher utility (because of lower value of risk + switching-cost) than any other strategy $\overline{g}^{'}$. The complete proof is provided in Appendix~\ref{app:thm-2-decentralization}.
\end{proof}

We also discuss a scenario where joining the largest mining pool is the best response for a player $p_{1}$ joining the mining pool. If any new player joins the largest pool, then the PoW blockchain system will become centralized as more new miners join the system. We show in Theorem~\ref{thm:bakwaas} below that joining the largest mining pools is the equilibrium strategy for any player $p_{1}$ joining the mining pool under certain conditions. The proof follows because this strategy's utility is higher than any other strategy given $p_{2}$ follows some $\overline{f}$. The complete proof is provided in Appendix~\ref{app:thm-bakwaas}.

\begin{Theorem}\label{thm:bakwaas}
    A PoW blockchain system with memoryless BRM $\Gamma_{memoryless}$ and game $\mathcal{G}(\Gamma_{memoryless})$ for player $p_{1}$ characterized by $\theta_{1} \in \mathbb{R}_{> 0}\times\mathbb{N}$ has equilibrium strategy $\overline{g}^{*}$ where $g_{i}^{*} = 1$ for $i = arg\max_{i \in \{1,2,\ldots,k\}} f_{i}$ if $c \geq \underline{c}= \frac{b\cdot R_{block}\cdot M_{1}\cdot p}{M_{2}\cdot D_{min}}$.
\end{Theorem}

For complete decentralization of the system, our goal is to ensure that mining pools are not formed. This can happen if, for any player $p_{1}$ joining the system, solo mining is incentivized over joining any subset of mining pools. For memoryless BRMs, we show the impossibility of constructing such BRMs that incentivize solo mining over pool mining for any value of $\rho > 1$ through Theorem~\ref{thm:impossibility}.

\begin{Theorem}\label{thm:impossibility}
    It is impossible to construct a memoryless BRM $\Gamma$ such that equilibrium strategy for $p_{1}$ characterized by $\theta_{1} = (M_{1},\rho)$ (for $\rho > 1$) in game $\mathcal{G}(\Gamma)$ is solo mining i.e. $\overline{g}^{*}$ such that $g^{*}_{0} = 1$.
\end{Theorem}
\begin{proof}
    The proof follows by showing existence of another strategy $\overline{g}^{'}$ which obtains higher utility for any given $\overline{f}$ than the solo mining strategy. The complete proof is provided in Appendix~\ref{app:thm-impossibility}.
\end{proof}

\begin{figure}
    \centering
    \begin{tikzpicture}
        \draw[red,very thick] (-2,0) -- (2,0) node[anchor=south east] {$c \geq \underline{c} = \frac{b\cdot R_{block}\cdot M_{1}\cdot p}{M_{2}\cdot D_{min}}$};
        \draw[red,very thick] (-2,0) -- (2,0) node[anchor=north east] {$\gets$ centralized};
        \draw[black] (2,0) -- (2,0.5) ;
        \draw[blue,very thick] (2,0) -- (4,0);   \draw[blue,very thick] (2,0) -- (3,0) node[anchor=south east] {open};
        \draw[blue,very thick] (3,0) -- (4,0); 
        \filldraw [brown] (4,0) circle (2 pt) node[anchor=north] {decentralized};
        \draw[black] (4,0) -- (4,0.5) node[anchor=south] {\color{brown}{$c = 0$}};
    \end{tikzpicture}
    \caption{Bounds on switching cost parameter for $\Gamma_{memoryless}$}
    \label{fig:bounds}
\end{figure}
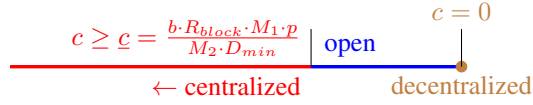

\paragraph{Need for Retentive BRMs} Lemma~\ref{thm:2-decentralization} and Theorem~\ref{thm:bakwaas} help us find bound on $c$ for which system is centralized, or $\rho$-Decentralized. We exhibit these bounds in Figure~\ref{fig:bounds}. We also make an interesting observation that $\underline{c}$ decreases as $M_{2}$ increases. Therefore, as the mining power of PoW blockchain system increases (which is already happening~\cite{hashrateIncrease}), the system's new miners are more and more prone towards joining the largest mining pool. This fact, in addition to Theorem~\ref{thm:impossibility} makes us realize the need for a retentive BRM where the equilibrium strategy for miners is solo mining rather than joining or forming a mining pool. 

\subsection{Retentive BRMs ($\Gamma_{fruitchain}$)}\label{ssec:retentive}

In this section, we (i) provide construction of retentive BRM used in Fruitchain Blockchain~\cite{Rafael2017Fruitchains}, (ii) show that fruitchain reduces risk for solo miners (iii) show that fruitchains also fail to disincentivize the formation of mining pools. 

\ssubsection{Fruitchains~\cite{Rafael2017Fruitchains}} It is a PoW-based blockchain that uses the concept of partial blocks to reduce risk in mining. Partial blocks have lesser target difficulty and are therefore easier to mine. A partial block refers to the most recent full block (blocks with higher target difficulty) $B_{k}$ as their parent. When a partial block (referencing parent $B_{k}$) is mined, its status is \emph{unconfirmed} (not a part of the ledger) and is represented as $\beta^{k}_{i}$ (for the $i^{th}$ unconfirmed partial block referencing $B_{k}$). These partial blocks are confirmed when a full block references them. That is, for some $r > k$, block $B_{k}$ references $\beta^{k}_{i}$, then the partial block is considered confirmed (part of the blockchain) and represented as $\gamma_{i}^{k}$. A partial block gets rewarded when it's confirmed. The currently mined block can only confirm partial blocks that were mined within the $z$ most recent full blocks mined. Partial blocks get rewarded $R_{partial}$ and full blocks get rewarded $R_{full}$. The parameters $R_{partial}, R_{full}, T_{full}$ (target difficulty of the full block), $T_{partial}$ (target difficulty of the partial block) and $z$ are set by the system. 

The current mined (full-block) is $B_{k}$ then it's history is represented by 
\[
    \begin{aligned}
        \mathcal{H} := & \big(\{B_{k-1},B_{k-2},\ldots,B_{1}\}, \{\overline{\gamma}^{k-2},\overline{\gamma}^{k-3},\ldots,\overline{\gamma}^{2}\}, \\
        & \{\overline{\beta}^{k-1},\overline{\beta}^{k-2},\ldots,\overline{\beta}^{k-z}\}\big)
    \end{aligned}
\]

\ssubsection{Game Formulation} Consider the set of transactions confirmed in block $B_{k}$ be contained in set $\mathcal{Z}_{k}$. Also, consider the indicator function $\mathbb{I}$ such that $\mathbb{I}(s_{1},\gamma) = 1$ if some block $\gamma$ is mined by player $p_{1}$ (having pub-key $s_{1}$) and $0$ otherwise. Therefore, $\Gamma_{fruitchain}$ for block $B_{k}$ and history $\mathcal{H}$ is represented for player $p_{1}$ as
\[
    \Gamma_{fruitchain}^{s_{1}}(B_{k},\mathcal{H}_{k}) = R_{full}\mathbb{I}(s_{1},B_{k}) + R_{partial}\sum_{\gamma \in \mathcal{Z}_{k}} \mathbb{I}(s_{1},\gamma) 
\]
For ease of notation, we represent $\Gamma_{fruitchain,k}^{s_{1}}(B_{k},\mathcal{H}_{k}) = \Gamma_{fruitchain,k}^{s_{1}}$. Similarly, $\mathbb{I}(pk_{i},\gamma)$ is the indicator function representing if a block $\gamma$ was mined by pool $i$ with pub-key $pk_{i}$. The game using memoryless BRM similar to fruitchains is represented as $\mathcal{G}(\Gamma_{fruitchain,k})$. From Equation~\ref{eqn:reward}, the reward for player $p_{1}$ in $\mathcal{G}(\Gamma_{fruitchain,k})$ in round $k$ under Assumption~\ref{assumption:approximate-mining-power} is represented by 

\begin{equation*}
    R_{k}^{fruitchain} = \begin{cases}
        \Gamma_{fruitchain,k}^{s_{1}}  & \text{with prob. }\frac{g_{0}M_{1}}{M_{2}} \\ 
        \Gamma_{fruitchain,k}^{pk_{i}}\psi_{i}(s_{1},\mathcal{M}(\overline{g},\overline{f})) & \text{with prob. }f_{i}\\
        0 & \text{otherwise}
    \end{cases}
\end{equation*}

\ssubsection{Analysis} The advantage of using retentive BRMs is they reduce risk for any new-coming player in case of solo mining. Consider two BRMs $\Gamma_{memoryless},\Gamma_{fruitchain}$ as previously described and $\mathcal{G}(\Gamma_{memoryless})$ and $\mathcal{G}(\Gamma_{fruitchain})$ be the games defined over the respective BRMs. Through Theorem~\ref{thm:retentive-better} we show that $\Gamma_{fruitchain}$ provides lower risk to player $p_{1}$ than $\Gamma_{memoryless}$.

\begin{Theorem}\label{thm:retentive-better}
        Retentive BRM $\Gamma_{fruitchain}$ provides lower risk than $\Gamma_{memoryless}$ towards solo mining strategy to player $p_{1}$ characterized by $\theta_{1} = (M_{1},\rho)$ for any $\rho > 1$.
\end{Theorem}

\begin{proof}
    To compare two different BRMs, we first impose the constraint that the reward distributed in one round is the same for both mechanisms. We then show that for a solo player $p_{1}$ that makes $M_{1}$ queries during which the rest of the system makes $M_{2}$ queries, we show that Utility is greater when BRM is $\Gamma_{fruitchain}$ than for $\Gamma_{memoryless}$. The complete proof is provided in Appendix~\ref{app:thm-retentive-better}.
\end{proof}

Although Fruitchains reduce the risk of solo mining, it also reduces the risk for mining pools. Therefore, centralization is an issue in Fruitchains, as demonstrated by Theorem~\ref{thm:fruitchains-centralized}. Note that centralization in fruitchains has been studied under a different model by~\cite{Stouka2023FruitchainsAnalysis} and the findings about fruitchains using our model align with their results. 

\begin{Theorem}\label{thm:fruitchains-centralized}
    Consider a player $p_{1}$ characterized by $\theta_{1} = (M_{1},\rho)$ joining the PoW blockchain system with BRM $\Gamma_{fruitchain}$. The player is always incentivized to join mining pools instead of solo mining. 
\end{Theorem}

\begin{proof}
    The proof follows by showing the existence of a strategy $\overline{g}$ where $g_{i} \neq 0$ for some $i \in \{1,2,\ldots,p\}$ (joining mining pools). We show this strategy has higher utility than solo mining, therefore a rational player $p_{1}$ is incentivized to join mining pools than solo mining. The complete proof is provided in Appendix~\ref{app:thm-fruitchains-centralized}.
\end{proof}

While fruitchains fail to achieve decentralization in the sense that solo mining is an equilibrium strategy for any new player joining the system, $\Gamma_{fruitchain}$ gives us an important insight. It reduces the gap between the optimal strategy $\overline{g}^{*}$ and solo mining strategy $(1,0,0\ldots,0)$. In the following section, we present a novel retentive BRM \proname. Under this BRM, solo mining becomes DSE after some number of rounds for any player $p_{1}$. 

\section{$\proname:$ Optimal BRM}\label{sec:protocol}

In this section, we (i) propose \proname, a risk-optimal BRM, and (ii) show that under the game $\mathcal{G}(\Gamma_{\proname})$, solo mining is the equilibrium strategy, which means using $\Gamma_{\proname}$ leads to a decentralized PoW blockchain.

\paragraph{\proname} The proposed BRM distributes rewards obtained from a block equally among miners who are part of the PoW blockchain. The reward from block $B_{k}$ is distributed among miners of blocks $B_{1},B_{2}\ldots,B_{k}$ equally. This allows a player to obtain a reward proportional to its mining power in each round. Since mining is a random process, for any player joining at round $r_{0}$, we consider the expected number of blocks mined by the player to be proportional to its mining power after $r_{0} + T$ with very high probability. 

The implementation of this protocol can be through a special transaction that inputs two (1) block $B_{k}$, (2) $B_{q}$ (for $q < k$) that was mined by public-key $pk_{q}$ and (3) signature from $pk_{q}$ to transfer $R_{block}\frac{B_{k}}{k}$ to the output address. This implementation is feasible and can be carried out through a \emph{fork} in most of the existing PoW-based blockchains. 

Next, we model this BRM under $\mathcal{G}$ and show that after $T$ rounds from joining, $p_{1}$ performing solo mining is an equilibrium strategy.

\paragraph{Game Formulation} Consider a player $p_{1}$ joining the PoW blockchain with the BRM $\Gamma_{\proname}$. In round $k$ where the block mined is $B_{k}$ and the history $\mathcal{H}_{k} = \{B_{k-1},B_{k-2},\ldots,B_{1}\}$. The indicator function for block $B$ is $\mathbb{I}(s, B) = 1$ if $B$ is mined by a player with pub-key $s$ and $0$ otherwise. The reward $R_{k}$ is given by 
\[
    \begin{aligned}
        \Gamma_{\proname}^{s_{1}}(B_{k},\mathcal{H}_{k}) = \sum_{i=1}^{k} \frac{R_{block}}{k}\mathbb{I}(s_{1},B_{i})
    \end{aligned}
\]
For ease of notation we represent $\Gamma^{s_{1}}_{\proname}(B_{k},\mathcal{H}_{k}) = \Gamma_{\proname,k}^{s_{1}}$. The game is represented as $\mathcal{G}(\Gamma_{\proname})$ where the reward obtained by player $p_{1}$ characterized by $(M_{1},\rho)$ in round $k$ is given by $R_{k}$ under Assumption~\ref{assumption:approximate-mining-power}. The public key for player $p_{1}$ is $s_{1}$ and for pool $i$ is $pk_{i}$. 

\[
    \begin{aligned}
        R_{k}^{\proname} = \begin{cases}
            \Gamma_{\proname,k}^{s_{1}} & \text{with prob. } \frac{g_{0}M_{1}}{M_{2}}\\
            \Gamma_{\proname,k}^{pk_{i}}\psi_{i}(s_{1},\mathcal{M}(\overline{g},\overline{f})) & \text{with prob. } f_{i}\\
            0 & \text{otherwise}
        \end{cases}
    \end{aligned}
\]

\paragraph{Analysis} We now show that \proname\ leads to a decentralized PoW blockchain system. Towards this, we show using Lemma~\ref{lemma:proto-decentralized} that for any player $p_{1}$ with type $\theta_{1} = (M_{1},\rho)$, solo mining $\overline{g}^{sm} = (1,0,0,\ldots,0)$ is equilibrium strategy. The lemma follows by showing that after $T$ rounds, the expected utility (expectation over the randomness of the mining protocol) is higher for solo mining than any other mining strategy with a very high probability. The proof is provided in Appendix~\ref{app:lemma-proto-decentralized}.

\begin{Lemma}\label{lemma:proto-decentralized}
    In the game $\mathcal{G}(\Gamma_{\proname})$, following solo mining strategy $\overline{g}^{sm}$ is a equilibrium strategy after $T$  rounds from joining for player $p_{1}$ characterized by $\theta_{1} = (M_{1},\rho)$.
\end{Lemma}

Since following $\overline{g}^{sm}$ is DSE, then the protocol is strongly $\rho$-Decentralized for any $\rho$, in Theorem~\ref{thm:main-theorem}.

\begin{Theorem}\label{thm:main-theorem}
    A PoW blockchain using $\Gamma_{\proname}$ block reward mechanism is $\rho$-Decentralized for any $\rho > 1$.
\end{Theorem}
\begin{proof}
    From Lemma~\ref{lemma:proto-decentralized}, we have $\Gamma_{\proname}$ is such that solo mining is a dominant strategy. After player $p_{1}$ joins the system:
    \[
        \max_{i\in[1,p]} f_{i} > \max_{j \in \{1,2,\ldots,p\}} \frac{g_{j}M_{1} + f_{j}M_{2}}{M_{1} + M_{2}}
    \]
    Therefore, the protocol is $\rho$-Decentralized. 
\end{proof}
\section{Conclusion \& Future Work}
\label{sec:conclusion}
\paragraph{Conclusion} To conclude, we understood through the rational decision of a new miner (player $p_{1}$) joining the PoW blockchain system (player $p_{2}$) about the decentralization guarantees of the BRM offered by the PoW blockchain under a more general utility model. We show that while bitcoin-like memoryless BRMs do not ensure decentralized BRMs, retentive BRMs offer better performance (lower risk for solo miners). We then propose \proname, a retentive BRM with solo mining as an equilibrium strategy. Therefore, the resulting PoW system will also be $\rho$-Decentralized for any $\rho$.

\paragraph{Future Work} We analyzed centralization/decentralization guarantees for the given protocol are such that for all $\rho > 1$ and  $c \geq \underline{c}$ the PoW blockchain with $\Gamma_{memoryless}$ is centralized, and $c = 0$ it is $\rho$-Decentralized. However, there is an open gap $0 < c \leq \underline{c}$ where we don't know the guarantees of the protocol. Closing this gap is left for future work. In addition, analysis of other BRMs such as~\cite{2,21-11}, which uses different mining functions or smart contract functionality of another (external to the analyzed PoW system) blockchain, is left for future work. 

\newpage

\bibliographystyle{unsrtnat}

\newpage

\input{appendix.tex}

\end{document}

%% file: appendix.tex
\appendix

\section{Proof for Claim~\ref{claim:fair-rss}}\label{appendix:claim-fair-rss}

\begin{proof}
    Consider an arbitrary mining pool $i$ following RSS $\psi_{i}$. Given BRM is $\Gamma$ and $\mathcal{G}(\Gamma)$ is the corresponding game. We consider two cases (i) $\mathbb{E}[\psi_{i}(s_{1},\mathcal{M}_{k}^{i})] < \frac{g_{i}M_{1}}{g_{i}M_{1} + f_{i}M_{2}}$ and (ii) $\mathbb{E}[\psi_{i}(s_{1},\mathcal{M}_{k}^{i})] > \frac{g_{i}M_{1}}{g_{i}M_{1} + f_{i}M_{2}}$. 

    \noindent \underline{\textbf{Case 1}} $\left(\mathbb{E}[\psi_{i}(s_{1},\mathcal{M}_{k}^{i})] < \frac{g_{i}M_{1}}{g_{i}M_{1} + f_{i}M_{2}}\right)$: In this case, consider the expected reward for player $p_{1}$ from investing $g_{i} = x$ in pool $i$.
    \[
        \begin{aligned}
            \mathbb{E}[R_{k}] = & \frac{\Gamma(B_{k},\mathcal{H}_{k})g_{0}M_{1}}{M_{1}+M_{2}}\\ & + \sum_{j=1}^{p} \Gamma(B_{k},\mathcal{H}_{k})\psi_{j}(s_{1},\mathcal{M}_{k}^{j})\frac{f_{i}M_{2} + g_{i}M_{1}}{M_{1} + M_{2}}\\
            = & \frac{R_{block}g_{0}M_{1}}{M_{1}+M_{2}} + \Gamma(B_{k},\mathcal{H}_{k})\psi_{i}(s_{1},\mathcal{M}_{k}^{i}) \\ & + \sum_{j=1,j \neq i}^{p} \Gamma(B_{k},\mathcal{H}_{k})\psi_{j}(s_{1},\mathcal{M}_{k}^{j})\frac{f_{i}M_{2} + g_{i}M_{1}}{M_{1} + M_{2}} \\
            < & \frac{R_{block}(g_{0} + g_{i})M_{1}}{M_{1}+M_{2}} \\ & + \sum_{j=1,j \neq i}^{p} \Gamma(B_{k},\mathcal{H}_{k})\psi_{j}(s_{1},\mathcal{M}_{k}^{j})\frac{f_{i}M_{2} + g_{i}M_{1}}{M_{1} + M_{2}}
        \end{aligned}
    \]
    In the utility equation for player $p_{1}$ in the game $\mathcal{G}(\Gamma)$, coefficient $a$ for expected reward $\mathbb{E}[R_{k}]$ is such that the coefficients for risk and switching cost $b$ and $c$ are negligible wrt. $a$. Further, since the increase in $\mathbb{E}[R_{k}]$ is non-negligible by switching from strategy $\overline{g}^{2} = \left(g_{0},g_{1},\ldots,g_{i},\ldots,g_{p}\right)$ to $\overline{g}^{1} = \left(g_{0}+g_{i},g_{1},\ldots,0,\ldots,g_{p}\right)$, therefore player $p_{1}$ always prefers $\overline{g}^{1}$ over $\overline{g}^{2}$. Hence pool $i$ does not get any computing resource from any $p_{1}$ of arbitrary type $\theta_{1} \in \mathbb{R}_{> 0}\times\mathbb{N}$.

    \noindent \underline{\textbf{Case 2}} $\left(\mathbb{E}[\psi_{i}(s_{1},\mathcal{M}_{k}^{i})] > \frac{g_{i}M_{1}}{g_{i}M_{1} + f_{i}M_{2}}\right)$: This case is infeasible for the mining pool. This is because, if for any incoming miner $p_{1}$ we have $\psi_{i} > \frac{g_{i}M_{1}}{g_{i}M_{1} + f_{i}M_{2}}$, then summing over all the miners that are part of the system, we get $\psi_{i} > 1$. This is because for any miner part of the mining pool $i$, we established from Case 1 that $\psi_{i} \geq \frac{g_{i}M_{1}}{g_{i}M_{1} + f_{i}M_{2}}$. However, $\psi_{1} > 1$ means the mining pool is distributing more rewards than they obtain, which bankrupts the mining pool and therefore is not possible. 

    Having established both cases, we come to the conclusion that for any mining pool $i$, the RSS should be Fair according to the Definition~\ref{def:fairness} for the pool to survive and miners join the pool. 
\end{proof}

\section{Proof for Lemma~\ref{thm:2-decentralization}}\label{app:thm-2-decentralization}

\begin{proof}
    We consider a game $\mathcal{G}(\Gamma_{memoryless})$ and analyse the DSE for player $p_{1}$.  The utility for $p_{1}$ in any round $r$ given by Equation~\ref{eqn:utility-single-round}. Since we study $2$-decentralization, $p_{1}$ is of type $\theta_{1} = (M_{1},\rho)$ for any $M_{1} \in \mathbb{R}_{> 0}$. The utility for any round $r$ is,
    \[
        \begin{aligned}
            u_{i}\left(\overline{g},\overline{f};(\theta_{1},\theta_{2})\right) \approx &  a \sum_{i=1}^{p} \frac{g_{i}M_{1}}{M_{2}}R_{block} + \frac{g_{0}M_{1}}{M_{2}}R_{block} \\
            & + (R_{block})^{\rho}\sum_{i=1}^{p} \frac{g_{i}^{\rho}}{f_{i}^{\rho-1}} + D(\overline{g})
        \end{aligned}
    \]
    We have due to Assumption~\ref{assumption:approximate-mining-power}, the probability of a miner $p_{1}$ solo-mining $\frac{g_{0}M_{1}}{M_{2}}$ is also very small compared to her mining as part of a pool. We first show that $\overline{g} = \overline{f}$ is the minima for risk -- $\mathbb{E}[R_{block}^{2}]$. Consider the risk when $\overline{g} = \overline{f}$.
    \[
        \left(\mathbb{E}[R_{block}^{\rho}]\right)^{1/\rho} = R_{block}\left(\sum_{i=1}^{p} \frac{g_{i}^{\rho}}{f_{i}}\right)^{1/\rho} = R_{block}\left(\sum_{i=1}^{p} f_{i}\right)
    \]
    Now, consider for any $\overline{g}^{'} \neq \overline{f}$. We divide $g_{i} \in \overline{g}^{'}$ into three parts $X = \{ i | g_{i} > f_{i}, g_{i} \in \overline{g}^{'}\}, Y = \{ i | g_{i} < f_{i}, g_{i} \in \overline{g}^{'}\}, $ and $Z = \{ i | g_{i} = f_{i}, g_{i} \in \overline{g}^{'}\}$. The risk for $\overline{g}^{'}$ is given as 
    \[
        \begin{aligned}
            \left(\mathbb{E}[R_{block}^{\rho}]\right)^{1/\rho} & = R_{block}\left(\sum_{l\in X} \frac{g_{l}^{\rho}}{f_{l}^{\rho-1}} + \sum_{m\in Y} \frac{g_{m}^{\rho}}{f_{m}^{\rho-1}} + \sum_{n\in Z} \frac{g_{n}^{\rho}}{f_{n}^{\rho-1}}\right)^{1/\rho}\\
            & = R_{block}\left(\sum_{l\in X} \frac{(f_{l} + \tau_{l})^{\rho}}{f_{l}^{\rho-1}} + \sum_{m\in Y} \frac{(f_{m} - \tau_{m})^{\rho}}{f_{m}^{\rho-1}} + \sum_{n\in Z} \frac{f_{n}^{\rho}}{f_{n}^{\rho-1}}\right)^{1/\rho}\\
            & \geq R_{block}\left(\sum_{i=1}^{p} f_{i} + \rho(\sum_{l \in X} \tau_{l} - \sum_{m \in Y} \tau_{m})\right)^{1/\rho}\\
            & = R_{block}^{\rho}\left(\sum_{i=1}^{p} f_{i}\right)^{1/\rho}
        \end{aligned}
    \]
    We use the lower bound $x^{v} - v\cdot x^{v-1}y \leq (x - y)^{v}$ and $\sum_{l \in X} \tau_{l} = \sum_{m \in Y} \tau_{m}$ for obtaining the above inequality. We have obtained that if risk is $r_1 = (\mathbb{E}[R_{block}^{\rho}])^{1/\rho}$ when $\overline{g} = \overline{f}$ and $r_2 = (\mathbb{E}[R_{block}^{\rho}])^{1/\rho}$  when $\overline{g} \neq \overline{f}$, then $r_1 \leq r_2$. Hence, strategy $\overline{g} = \overline{f}$ is risk-otpimal for any $\rho$.   
    
    Since, (i) expected reward is the same for all strategies $\overline{g}$, (ii) risk is minimized for strategy $\overline{g} = \overline{f}$ and (iii) $c = 0$,  the strategy $\overline{g} = \overline{f}$ gives the optimal utility for $p_{1}$.     
\end{proof}

\section{Proof for Theorem~\ref{thm:bakwaas}}\label{app:thm-bakwaas}
\begin{proof}
    To prove this theorem, we consider the utility from two strategies (1) $\overline{g}^{1}$ where $g^{1}_{i} = 1$ for $i = arg\max_{j\in [1,p]} f_{j}$ and $0$ otherwise, (2) any other $\overline{g}^{2} \in S_{\alpha}$ in $\alpha$-Marginal Strategy space. Wlog. we take $f_{1} \geq f_{2} \geq \ldots \geq f_{p}$.  Consider the utility of player $p_{1}$ in Game $\mathcal{G}(\Gamma_{memoryless})$ following strategy $\overline{g}^{1}$ after Assumption~\ref{assumption:approximate-mining-power} as
    \[
        U_{1} = u_{1}\left(\overline{g}^{1},\overline{f};(\theta_{1},\theta_{2})\right) \approx \frac{aR_{block}M_{1}}{M_{2}} - \left(\frac{b\cdot R_{block}M_{1}}{M_{2}f_{1}^{(\rho-1)/\rho}}\right) - c\cdot D(\overline{g}^{1})
    \]
    Similarly, consider the utility for some strategy $\overline{g}^{2} = (g_{1},g_{2},\ldots,g_{p})$ after Assumption~\ref{assumption:approximate-mining-power} is, 
    \[
        \begin{aligned}
            U_{2} = u_{1}\left(\overline{g}^{2},\overline{f};(\theta_{1},\theta_{2})\right) & \approx \frac{aR_{block}M_{1}}{M_{2}} - c\cdot D(\overline{g}^{2})\\
            & - \frac{b\cdot R_{block}M_{1}}{M_{2}}\left(\sum_{j=1}^{p}\frac{g^{\rho}_{j}}{f^{\rho-1}_{j}}\right)^{1/\rho}
        \end{aligned}
    \]
    We now take the value $U_{1} - U_{2}$ which gives us
    \[
    \begin{aligned}
        U_{1} - U_{2} = & c\cdot(D(\overline{g}^{2}) - D(\overline{g}^{1})) \\ & - \frac{b\cdot R_{block}M_{1}}{M_{2}}\left(\frac{1}{f_{1}^{(\rho - 1)/\rho}}
         - (\sum_{j=1}^{p} \frac{g_{j}^{\rho}}{f_{j}^{\rho-1}})^{1/\rho}\right)\\ 
         \geq & c \cdot(D(\overline{g}^{2}) - D(\overline{g}^{1})) \\ & - 
         \frac{b\cdot R_{block}M_{1}}{M_{2}}\left(p^{(\rho - 1)/\rho}
         - 1\right)\\ 
         \geq & c \cdot D_{min} - \frac{b \cdot R_{block}M_{1}p^{(\rho-1)/\rho}}{M_{2}} \\
         \geq & c \cdot D_{min} - \frac{b\cdot R_{block}M_{1}p}{M_{2}} \geq 0
    \end{aligned}
    \]
    The first inequality comes from the lower bound on $\left(\sum_{i}\frac{g_{i}^{\rho}}{f_{i}^{\rho-1}}\right) = 1$ from proof of Theorem~\ref{thm:2-decentralization} in Appendix~\ref{app:thm-2-decentralization} and $\frac{1}{f_{1}} \leq p$ by simple algebra (note $f_{1} \geq f_{2} \geq \ldots \geq f_{p}$). Second inequality comes since $\overline{g}^{1},\overline{g}^{2} \in S_{\alpha}$ are two distinct points on $\alpha$-Marginal Strategy Space, therefore $|D(\overline{g}^{2}) - D(\overline{g}^{1})| \geq D_{min}$. However, by property \textbf{P1} of the function $D$, we have $D(\overline{g}^{2}) \geq D(\overline{g}^{1})$. Therefore, $D(\overline{g}^{2}) - D(\overline{g}^{1}) \geq D_{min}$. The last inequality comes from $c\cdot M_{2}\cdot D_{min} \geq b\cdot R_{block}\cdot M_{1}p$.

    Therefore, $U_{1} \geq U_{2}$. This means strategy $\overline{g}^{1}$ is Equilibrium from Definition~\ref{def:eq} and therefore the system becomes centralized as all the new miners are incentivized to join the largest mining pool.
\end{proof}

\section{Proof for Theorem~\ref{thm:impossibility}}\label{app:thm-impossibility}
\begin{proof}
    To prove this theorem, we consider payoff in any arbitrary round $k$ where the block mined is $B_{k}$ and history is $\mathcal{H}_{k}$. The memoryless BRM is $\Gamma_{memoryless}$ and game is $\mathcal{G}(\Gamma_{memoryless})$. Wlog. consider $f_{1} \geq f_{2} \geq \ldots f_{p}$. To show that any memoryless BRM will lead to the formation of mining pools, we compare utility for player $p_{1}$ of any type $\theta_{1} \in \mathbb{R}_{>0}\times\mathbb{N}$ for two strategies:
    \begin{itemize}
        \item $\overline{g}^{lm}$ -- when $p_{1}$ dedicates mining power largest mining pool. Here $g^{lm}_{i} = 1 \; i = arg\max_{j \in [1,k]} f_{j}$. Utility for this strategy is $U_{lm} = u_{1}\left(\overline{g}^{lm},\overline{f};(\theta_{1},\theta_{2})\right)$.
        \item  $\overline{g}^{sm}$ -- when miner dedicates their entire mining power to solo mining. In this case, $g^{sm}_{0} = 1$. Utility for this strategy is $U_{sm} = u_{1}\left(\overline{g}^{sm},\overline{f};(\theta_{1},\theta_{2})\right)$.
    \end{itemize}
    After Assumptions~\ref{assumption:approximate-mining-power} utility $U_{lm}$ and $U_{sm}$ are given as
    \[
        U_{lm} = \frac{a\cdot R_{block}\cdot M_{1}}{M_{2}} - \frac{b\cdot R_{block}\cdot M_{1}}{M_{2}}\left(f_{1}\right)^{1/\rho} - c\cdot D(\overline{g}^{lm})
    \]
    \[
        U_{sm} = \frac{a\cdot R_{block}\cdot M_{1}}{M_{2}} - \frac{b\cdot R_{block}\cdot M^{1/\rho}_{1}}{M^{1/\rho}_{2}} - c\cdot D(\overline{g}^{sm})
    \]
    We make the simple observation that (i) $D(\overline{g}^{sm}) \neq D(\overline{g}^{lm})$ and (ii) from Assumption$\frac{M_{1}}{M_{2}} << 1 \implies \left(\frac{M_{1}}{M_{2}}\right)^{1/\rho} > \frac{p\cdot M_{1}}{M_{2}} > \frac{f_{1}\cdot M_{1}}{M_{2}}$. Therefore, $U_{lm} > U_{sm}$ for $\rho > 1$. With this, we conclude that $\overline{g}^{sm}$ (solo-mining strategy) is always dominated by $\overline{g} = \overline{f}$ which means mining pools will exist in Memoryless BRMs. 
\end{proof}

\section{Proof for Theorem~\ref{thm:retentive-better}}\label{app:thm-retentive-better}

\begin{proof}
    To prove this theorem, we consider the ratio of the risk ($\rho^{th}$ moment) to the expected payoff (raised to $\rho^{th}$ power) for any player $p_{1}$ that is mining in the PoW blockchain in two cases -- $\Gamma_{memoryless}$ and $\Gamma_{fruitchain}$. The Target difficulty for full blocks in Fruitchains as well as Bitcoin is $T_{full}$ and for partial blocks in Fruitchains is $T_{partial}$. Since it is easier to mine a partial block, we have $T_{full} < T_{partial}$. The number of partial blocks mined in expectation, for every full block in fruitchains is, therefore, $\frac{T_{partial}}{T_{full}}$\footnote{The result follows from simple probability theory, but details can be found in~\cite{Rafael2017Fruitchains}}. In addition, reward given on mining a full block in fruitchains should be equal to the reward from memoryless BRMs. Therefore, 
    \begin{equation}\label{eqn:app-cross-mechanism-expected-reward}
        R_{block} = R_{full} + \frac{T_{partial}}{T_{full}}\cdot R_{partial}
    \end{equation}
    We now calculate $(\mathbb{E}[R_{k}^{\rho}])^{1/\rho}$ for both Memoryless BRMs which we call $Risk_{memoryless}$ and fruitchains which we call $Risk_{fruitchain}$. 
    
    \smallskip \noindent \underline{\textbf{Case 1 ($\Gamma_{memoryless}$):}} The game being played is $\mathcal{G}(\Gamma_{memoryless})$  and player $p_{1}$ has type $\theta_{1} = (M_{1},\rho)$ and player $p_{2}$ has type $\theta_{2} = (M_{2},A)$. This means, by the time system makes $M_{2}$ queries for mining, the player $p_{1}$ makes $M_{1}$ queries. The probability of a query made by $p_{1}$ in mining a block is $\gamma_{1} = \frac{T_{full}}{2^{\lambda}}$. The reward obtained is $R_{block}$. Therefore, 
    \[
        Risk_{memoryless} = (E[R_{k}^{\rho}])^{1/\rho} = \gamma_{1}^{1/\rho}\cdot R_{block}
    \]
    
    \smallskip \noindent \underline{\textbf{Case 2 ($\Gamma_{fruitchain}$):}} The game being played is $\mathcal{G}(\Gamma_{fruitchain})$ and the $p_{1}$ has type $\theta_{1} = (M_{1},\rho)$ and player $p_{2}$ has type $\theta_{2} = (M_{2},A)$. Probability of mining in one query a full block is $\gamma_{1} = \frac{T_{full}}{2^{\lambda}}$ and for partial block is $\gamma_{2} = \frac{T_{partial}}{2^{\lambda}}$. The reward obtained is $R_{full}$ and $R_{partial}$ respectively. Therefore, 
    \[
        Risk_{fruitchain} = (E[R_{k}^{\rho}])^{1/\rho} = \left(\gamma_{1}R^{\rho}_{full} + \gamma_{2}R^{\rho}_{partial}\right)^{1/\rho}
    \]
    Now we calculate the ratio of risks below
    \[
        \begin{aligned}
            \frac{Risk_{fruitchain}}{Risk_{memoryless}} & = \left(\frac{\gamma_{1}R^{\rho}_{full} + \gamma_{2}R^{\rho}_{partial}}{R_{block}^{\rho}\gamma_{1}} \right)^{1/\rho}\\
            \leq & \left(1 - \frac{R_{partial}(\gamma_{1}T^{\rho}_{partial} - \gamma_{2}T_{full}^{\rho})}{R_{block}^{\rho}T_{full}^{\rho}\gamma_{1}}\right)\\
            < & 1 
        \end{aligned}
    \]
    The first inequality is obtained from $R_{full} = R_{block} - \frac{T_{partial}}{T_{full}}R_{partial}$ and second inequality is obtained because for any $\rho > 1$ we have $\gamma_{1}T^{\rho}_{partial} - \gamma_{2}T_{full}^{\rho} > 0$. Therefore, we show that $Risk_{fruitchains} < Risk_{memoryless}$ for solo miners. Since strategy $\overline{g}$ is the same in both cases, the switching cost is the same for both Memoryless BRMs and Fruitchains. In addition, the expected reward remains the same due to Equation~\ref{eqn:app-cross-mechanism-expected-reward} in the Appendix. Hence, the utility for solo miners is higher in fruitchains than memoryless BRMs for the same expected block reward. 
\end{proof}

\section{Proof for Theorem~\ref{thm:fruitchains-centralized}}\label{app:thm-fruitchains-centralized}

\begin{proof}
    Consider a player $p_{1}$ characterized by $\theta_{1} = (M_{1},\rho)$ and $p_{2}$ characterized by $(M_{2},A)$. The game being played is $\mathcal{G}(\Gamma_{fruitchain})$. Player $p_{1}$ plays after observing $\overline{f}$ (strategy chosen by $p_{2}$). Wlog. consider $f_{1} \geq f_{2} \geq \ldots f_{p}$. Consider a round $k$ where block $B_{k}$ is mined. We need to show that joining mining pools dominates the strategy to perform solo mining. We represent the strategies: 
    \begin{itemize}
        \item $\overline{g}^{sm} = (1,0,0,\ldots,0)$ as strategy for \emph{solo-mining} with $g^{sm}_{0} = 1$. The utility corresponding to solo-mining is $U_{sm}$.
        \item $\overline{g}^{lm} = (0,1,0,0,\ldots,0)$ as the strategy where the miner joins the \textbf{l}argest \textbf{m}ining pool. The corresponding utility is $U_{lm}$.
    \end{itemize}
    Our goal is to show that  $U_{lm} > U_{sm}$ for all instances of $\mathcal{G}(\Gamma_{fruitchain})$ with different characterizations of $p_{1}$ and $p_{2}$. We first make the observation that $D(\overline{g}^{lm}) = D(\overline{g}^{sm})$ because in either case, there is no two pools to switch between. Now, $U_{sm}$ is given by
    \[
    \begin{aligned}
        U_{sm} & = &  a\cdot\mathbb{E}[R_{k}] & - b\cdot(\mathbb{E}[R^{\rho}_{k}])^{1/\rho} - c\cdot D(\overline{g}^{sm})\\
        & = &  \frac{a\cdot\Gamma_{fruitchain,k}M_{1}}{M_{2}} & - b\cdot\Gamma_{fruitchain,k}\cdot\left(\frac{M_{1}}{M_{2}}\right)^{1/\rho} - c\cdot D(\overline{g}^{sm})
    \end{aligned}
    \]
    Similarly, we get $U_{lm}$ as 
    \[
    \begin{aligned}
        U_{lm} & = &  a\cdot\mathbb{E}[R_{k}] & - b\cdot(\mathbb{E}[R^{\rho}_{k}])^{1/\rho}  - c\cdot D(\overline{g}^{lm})\\
        & = &  \frac{a\cdot\Gamma_{fruitchain,k}M_{1}}{M_{2}} & - b\cdot\Gamma_{fruitchain,k}\cdot\frac{M_{1}}{M_{2}}\frac{1}{f_{1}^{(\rho-1)/\rho}}  - c\cdot D(\overline{g}^{lm}) \\
        & \leq &  \frac{a\cdot\Gamma_{fruitchain,k}M_{1}}{M_{2}} & - b\cdot\Gamma_{fruitchain,k}\cdot\frac{M_{1}}{M_{2}} - c\cdot D(\overline{g}^{lm})
    \end{aligned}
    \]
    Therefore, taking $U_{lm} - U_{sm}$ and since $D(\overline{g}^{sm}) = D(\overline{g}^{lm})$, we get 
    \[
        U_{lm} - U_{sm} = b\Gamma_{fruitchain,k}\left(\frac({M_{1}}{M_{2}})^{1/\rho} - \frac{M_{1}}{M_{2}}\right) > 0
    \]
    Inequality follows since for $M_{1} < M_{2}$ (from Assumption~\ref{assumption:approximate-mining-power}) we have $\frac{M_{1}}{M_{2}} < \left(\frac{M_{1}}{M_{2}}\right)^{1/\rho}$ for any $\rho > 1$. Thus, joining the largest pool \emph{strictly dominates} solo mining for any instance of $\mathcal{G}(\Gamma_{fruitchains})$.
\end{proof}

\section{Proof for Lemma~\ref{lemma:proto-decentralized}}\label{app:lemma-proto-decentralized}
\begin{proof}
    Consider player $p_{1}$ joining the protocol in round $r_{0}$. After $T$ rounds, consider any round $k > r_{0} + T$. The block proposed in this round is $B_{k}$ and $\mathcal{H}_{k}$ be the history. We consider two strategies: 
    \begin{itemize}
        \item $\overline{g}^{sm}$ -- \emph{solo mining} where $g^{sm}_{0} = 1$. The utility corresponding to this strategy is $U_{sm}$.
        \item $\overline{g}^{mp}$ -- $p_{1}$ invests in \emph{mining pools} according to $g^{mp} \in S_{\alpha}$. The utility corresponding to this strategy is $U_{mp}$.
    \end{itemize}
    The reward from the block $B_{k}$ for $p_{1}$ playing the solo mining strategy is given by the random variable $R_{k}^{sm}$. The probability that a block was mined by player $p_{1}$ under Assumption~\ref{assumption:approximate-mining-power} is $\frac{M_{1}}{M_{2}}$.
    \[
        \mathbb{E}[R_{k}^{sm}] = \sum_{i=1}^{k} \frac{M_{1}}{M_{2}}\cdot \frac{R_{block}}{k} = \frac{R_{block}\cdot M_{1}}{M_{2}} 
    \]
    The risk is given as $(\mathbb{E}[(R_{k}^{sm})^{\rho}])^{1/\rho}$. We calculate this risk as 
    \[
        (\mathbb{E}[(R_{k}^{sm})^{\rho}])^{1/\rho} = \frac{R_{block}\cdot M_{1}}{M_{2}}\left(\frac{M_{1}}{M_{2}}\right)^{1/\rho}
    \]
    Similarly, consider random variable for strategy $\overline{g}^{mp}$ is $R_{k}^{mp}$. 
    \[
        \mathbb{E}[R_{k}^{mp}] = \sum_{i=1}^{k} \sum_{j=1}^{p} f_{j}\frac{R_{block}}{k}\frac{g_{j}M_{1}}{f_{j}M_{2}} = \sum_{j=1}^{p} \frac{R_{block}g_{j}M_{1}}{M_{2}} = \frac{R_{block}\cdot M_{1}}{M_{2}}
    \]
    The risk is given as $(\mathbb{E}[(R_{k}^{mp})^{\rho}])^{1/\rho}$. We calculate the risk as 
    \[
    \begin{aligned}
        (\mathbb{E}[(R_{k}^{mp})^{\rho}])^{1/\rho} & = \left(\sum_{i=1}^{k} \sum_{j=1}^{p} f_{j}\left(\frac{R_{block}}{k}\frac{g_{j}M_{1}}{f_{j}M_{2}}\right)^{\rho}\right)^{1/\rho}\\
        & = \frac{R_{block}M_{1}}{M_{2}} \left(\sum_{j=1}^{p} f_{j}\left(g_{j}\right)^{\rho}\right)^{1/\rho}
    \end{aligned}
    \]
    Difference in utility $U_{sm} - U_{mp}$ is given by 
    \[
    \begin{aligned}
        U_{sm} - U_{mp} & = & a(\mathbb{E}[R_{k}^{sm}] - \mathbb{E}[R_{k}^{mp}]) \\ & & - b((\mathbb{E}[(R_{k}^{sm})^{\rho}])^{1/\rho} - (\mathbb{E}[(R_{k}^{mp})^{\rho}])^{1/\rho}) \\ & & - c\cdot (D(\overline{g}^{sm}) - D(\overline{g}^{mp})) 
    \end{aligned}
    \]
    We know $D(\overline{g}^{sm}) - D(\overline{g}^{mp}) \leq 0$ because solo mining has no switching cost. In addition, the expected payoff is the same in both cases. Therefore,  
    \[
    \begin{aligned}
        U_{sm} - U_{mp} & > \frac{R_{block}\cdot M_{1}\cdot b}{M_{2}}\left( \left(\sum_{j=1}^{p} f_{j}\left(g_{j}\right)^{\rho}\right)^{1/\rho} - \left(\frac{M_{1}}{M_{2}}\right)^{1/\rho} \right) \gtrapprox 0
    \end{aligned}
    \]
    The last inequality is obtained because from Assumption~\ref{assumption:approximate-mining-power} $\frac{M_{1}}{M_{2}}$ is very small. Therefore, we obtain $U_{sm} \gtrapprox U_{mp}$ which means solo mining is more profitable (or in worse cases, negligibly less profitable) than mining pools when BRM is $\Gamma_{\proname}$.
\end{proof}